\documentclass[conference,letterpaper]{IEEEtran}

\usepackage[utf8]{inputenc} 
\usepackage[T1]{fontenc}
\usepackage{url}
\usepackage{aliascnt}
\usepackage{hyperref}
\usepackage[cmex10]{amsmath} 
\usepackage{amssymb}
\usepackage{amsthm}
\usepackage[style=ieee]{biblatex}

\usepackage{verbatim}
\usepackage{xcolor}
\usepackage{mathrsfs}
\usepackage[normalem]{ulem}
\usepackage{colortbl} 

\newtheorem{theorem}{Theorem}

\newaliascnt{lemma}{theorem}

\aliascntresetthe{lemma}

\newaliascnt{proposition}{theorem}
\newtheorem{proposition}[proposition]{Proposition}
\aliascntresetthe{proposition}

\newaliascnt{corollary}{theorem}
\newtheorem{corollary}[corollary]{Corollary}
\aliascntresetthe{corollary}

\newaliascnt{example}{theorem}
\newtheorem{example}[example]{Example}
\aliascntresetthe{example}

\newaliascnt{definition}{theorem}
\newtheorem{definition}[definition]{Definition}
\aliascntresetthe{definition}

\newaliascnt{claim}{theorem}
\newtheorem{claim}[claim]{Claim}

\aliascntresetthe{claim}

\newtheorem{problem}{Problem}

\newtheorem{construction}{Construction}

\newcommand{\run}[1]{\rho \! \left( #1 \right)}
\newcommand{\setX}{\mathcal{X}\!\left( \mathbf{x} \right)}
\newcommand{\maxcar}[2]{A_n^\mathbb{#1}\left(M, #2 \right)}

\IEEEoverridecommandlockouts

\begin{document}
\title{\textbf{Coding for Composite DNA to Correct Substitutions, Strand Losses, and Deletions}} 



\author{%
  \IEEEauthorblockN{\textbf{Frederik Walter}\IEEEauthorrefmark{1},
                    \textbf{Omer Sabary}\IEEEauthorrefmark{2},
                    \textbf{Antonia Wachter-Zeh}\IEEEauthorrefmark{1},
                    and \textbf{Eitan Yaakobi}\IEEEauthorrefmark{2}
                    }
  \IEEEauthorblockA{\IEEEauthorrefmark{1}%
                   Institute for Communication Engineering,  Technical University of Munich, Germany,
                    }
  \IEEEauthorblockA{\IEEEauthorrefmark{2}%
                    Faculty of Computer Science, Technion-Israel Institute of Technology, Haifa, Israel
                    \\
                    Emails: \{frederik.walter, antonia.wachter-zeh\}@tum.de
                    \{omersabary, yaakobi\}@cs.technion.ac.il}
                    \vspace*{-1cm}
\thanks{This project has received funding from the European Union’s Horizon Europe research and innovation programme under grant agreement No. 101115134}
}

\maketitle

\begin{abstract}
\emph{Composite DNA} is a recent method to increase the base alphabet size in DNA-based data storage.
This paper models synthesizing and sequencing of \emph{composite DNA} and introduces coding techniques to correct substitutions, losses of entire strands, and symbol deletion errors.  
Non-asymptotic upper bounds on the size of codes with $t$ occurrences of these error types are derived. Explicit constructions are presented which can achieve the bounds. 
\end{abstract}
\vspace{-0.2cm}
\section{Introduction}
Data storage on DNA molecules is a promising approach for archiving massive data \autocite{church_2012_NextGenerationDigital, goldman_2013_PracticalHighcapacity, erlich_2017_DNAFountain, organick_2018_RandomAccess}.
In classical DNA storage systems, binary information is encoded into sequences consisting of the four DNA bases
$\{A, C, G, T\}$. The encoded sequences are used to generate DNA molecules called \emph{strands} using the biochemical process of DNA synthesis. The synthesized strands are stored together in a tube. To retrieve the binary information, the strand must be read via \emph{DNA sequencing} and decoded back into the binary representation. 
The synthesis and the sequencing procedures are error-prone, and with the natural degradation of DNA they introduce errors to the DNA strands. To ensure data reliability, the errors have to be corrected by algorithms and error-correcting codes (ECCs). 

Recently, to allow higher potential information capacity, \autocite{anavy_2019_DataStorage,choi_2019_HighInformation} introduced the \emph{composite DNA} synthesis method. In this method, the multiple copies created by the standard DNA synthesis method
are utilized to create \emph{composite DNA symbols}, defined by a mixture of DNA bases and their ratios in a specific position of the strands. By defining different mixtures and ratios, the alphabet can be extended to have more than 4 symbols. More formally, a composite DNA symbol in a specific position can be abstracted as a quartet of probabilities $\{p_A, p_C, p_G, p_T\}$, in which $p_X$, $0 \le p_X \le 1$, is the fraction of the base $X \in \{A, C, G, T\}$ in the mixture and $p_A+p_C+ p_G+  p_T =1$. Thus, to identify composite symbols it is required to sequence multiple reads and then to estimate $p_A, p_C, p_G, p_T$ in each position. 

Due to the unique structure of the alphabet symbols in this method, base-level errors can easily change the observed mixture of bases and their ratio, therefore changing the observed composite symbols. Moreover, in this setup, the inherent redundancy of the synthesis and sequencing processes (i.e., multiple copies per strand) cannot be used directly to overcome errors by a retrieval pipeline \autocite{bar-lev_2021_DeepDNA,sabary_2024_ReconstructionAlgorithmsa} and thus it is required to design ECCs specifically targeting this method. An extension of the composite method, in which the symbols are composed from short DNA fragments (known as \emph{shortmers}) was suggested in \autocite{preuss_2024_EfficientDNAbased, yan_2023_ScalingLogical}. Other coding and information theory problems related to composite DNA were studied in \autocite{kobovich_2023_MDABInputDistribution, preuss_2024_SequencingCoverage,cohen_2024_OptimizingDecoding}. 

The most related work to this paper was recently studied by Zhang et al. in \autocite{zhang_2022_LimitedMagnitudeError}. The authors initiated the study of error-correcting codes for composite DNA. They considered an error model for composite symbols, which assumes that errors occur in at most $t$ symbols, and their magnitude is limited by $\ell$. They presented several code constructions as well as bounds for this model. Our work proposes a different way to model the composite synthesis method and studies additional error models. To simplify the model, the results are presented for the binary base alphabet instead of the 4-ary. The errors discussed in this paper include substitution errors, deletions, insertions, and the loss of entire strands. We suggest code constructions for these models and study upper bounds on the code cardinality. Due to space limitations, missing proofs can be found in the appendix. 
\vspace{-0.3cm}
\section{Definitions and Problem Statement}\label{sec:defs}
\vspace{-0.3cm}
Our approach to modeling composite DNA is described as follows. For simplicity, our model assumes the composite symbols are created from a binary alphabet (compared to a $4$-ary alphabet), allowing us to have only two probabilities $p_0$ and $p_1$. To further simplify the model, we assume that exactly $M$ strands are synthesized, resulting in probabilities that are multiples of $\frac{1}{M}$.
Thus, in our model, data generated by composite DNA can be described in two forms. 
The first is called \emph{composite vector representation}, a length-$n$ vector over the alphabet $\{0, 1, \ldots, M\}$. The second form is a \emph{matrix representation}, in which the data is described by an $M\times n$ binary matrix. 
The matrix representation, therefore, explicitly represents the synthesized strands (which are the rows of the matrix) while the vector representation corresponds to the composite mixture. 
For $1\le j \le n$, the number of ones in the $j$-th column of the matrix sums to the value in the $j$-th position of the composite vector. Clearly, there is a one-to-many mapping between a composite vector and its corresponding matrices.
In this work, we assume that all the strands that compose the matrix representation of a composite vector are classified and clustered perfectly\footnote{This can be done by defining standard (non-composite) \emph{indices} in each of the strands. Then, using the indices it is possible to identify strands that relate to the same composite vector.}.   
Furthermore, it should be noted, as mentioned in the introduction, that the synthesis process produces a \emph{set} of strands, that are described in our model as a \emph{matrix}. This is done to order the strands and simplify the analysis while the same results can be achieved if one decides to work with sets rather than matrices. 

For positive integers $k,n$, let $[k,n] = \{ k, k+1,  \dots , n \}$. 
For a composite vector $\mathbf{x} = (x_1, \dots , x_n) \in [0,M]^n$, we denote by $\setX$ the set of all possible matrix representations of $\mathbf{x}$, where each single matrix representation is denoted by  $X\in\setX$. More formally, 
\vspace{-0.3cm}
\begin{equation*}
    \setX = \Big\{ X \in \{0,1\}^{M \times n} : \sum_{i=1}^M X_{i,j} = x_j, \forall j \in [1,n] \Big\}.
\end{equation*}
\\[-2ex]
The cardinality $ |\setX | $ and average cardinality $\text{E}_{\mathbf{x}}(|\setX |)$ of the set $\setX$ are given by 
\vspace{-0.15cm}
\begin{equation*}
  |\setX | = \prod_{j=1}^{n} \binom{M}{x_j}, 
  \quad
  \text{E}_\mathbf{x}(|\setX |) = \frac{2^{Mn}}{(M+1)^n}.
\end{equation*}
\\[-2ex]
Thus, every permutation of rows of $X$ is also in $\setX$ and even each permutation within each of the columns.

When we refer to a \emph{strand}, we refer to the respective row in the matrix representation $X$.
For a matrix $X \in \setX$, $X_i$ denotes the $i$-th row of $X$, and $X_{i,j}$ denotes the $j$-th element in the $i$-th row. If we add ($+$) or subtract ($-$) two binary matrices $X, Y$, all operations are done element-wise modulo~$2$. 

In this paper, we assume that errors are introduced to the matrix representation $X$ of composite vectors $\textbf{x}$, and the goal is to retrieve $\textbf{x}$. 
We will analyze five error types in this paper, which are defined in the remainder of the section. 

When discussing channel properties independent of the error type, we use $\mathbb{E}$ as a placeholder.
The channel output, i.e., the noisy version of $X$, is denoted by $R$ and is given to the decoder. 
We define $\mathcal{R}^\mathbb{E}$ as the set of all possible channel matrices $R$ that can be obtained from any composite vector $\mathbf{x} \in [0,M]^n$ when errors of type $\mathbb{E}$ occur. 
For readability, we refer to $R$ as matrices and use the notation even if some elements do not contain a symbol in some error scenarios. 
In some setups, it is useful to transform $R$ to composite vector representation by summing the ones in each column of $R$. In this case, we denote the resulting vector by $\mathbf{r}$. 

\begin{definition}
The error ball of radius $t$ of type $\mathbb{E}$, denoted by $B_t^\mathbb{E} (\mathbf{x}) \subseteq \mathcal{R}^\mathbb{E}$, is the set of all matrices which can be obtained by introducing any $t$ errors of type $\mathbb{E}$ in any of the matrix representations $X$ of $\mathbf{x}$. 
\end{definition}

\begin{definition}\label{def:substitution}
    Let $\mathbf{x} \in [0,M]^n$ be a composite vector with a possible matrix representation $X \in \{ 0,1 \}^{M \times n}$. $\mathbb{S}$ denotes the substitution error type. It is said that \textbf{$t$ substitution errors} occurred in $X$, if there exist $t$ tuples $(k_\ell, h_\ell), k_\ell \in [1,M], h_\ell \in [1,n], \ell \in [1, t]$  
    such that
    \vspace{-0.2cm}
    \begin{equation*}
      R_{i,j} = 
      \begin{cases}
          \overline{X}_{i,j} & \text{if } i=k_\ell, j=h_\ell, \forall \ell \in [1,t]
          \\
          X_{i,j} & \text{otherwise},
      \end{cases}
    \end{equation*}
    with $R \!\in\! \{ 0,1 \}^{M \times n}$ and $\!\overline{X}_{i,j}\!$~is~the~binary~complement~of~$X_{i,j}$.
    \\
    A code $\mathcal{C}_t^\mathbb{S} \subseteq [0,M]^n$ is called \textbf{$t$-substitution-correcting code} if for  any $\mathbf{c}, \mathbf{c}' \in \mathcal{C}_t^\mathbb{S}$, we have $B_t^\mathbb{S} (\mathbf{c}) \cap B_t^\mathbb{S} (\mathbf{c}') = \emptyset$.
    The maximum cardinality of a $t$-substitution-correcting code is denoted by $\maxcar{S}{t}$.
\end{definition} 

\begin{problem}\label{pro:substitution}
    Find the value of $\maxcar{S}{t}$ and $t$-substitution-correcting codes of cardinality $\maxcar{S}{t}$.
\end{problem}

\begin{definition}\label{def:strandloss}
    Let $\mathbf{x} \in [0,M]^n$ be a composite vector with a possible matrix representation $X \in \{ 0,1 \}^{M \times n}$. $\mathbb{L}$ denotes the strand loss error type. It is said that \textbf{$t$ strand losses} occurred in $X$, if there exist $t \in [1,M]$ indices $k_\ell \in [1,M], \ell \in [1,t]$, such that  $R \in \{ 0,1 \}^{(M-t) \times n}$ is a submatrix of $X$, obtained by removing the rows indexed by $k_{\ell}$.
    A code $\mathcal{C}_t^\mathbb{L} \subseteq [0,M]^n$ is called a  \textbf{$t$-strand-loss-correcting code} if for any $\mathbf{c}, \mathbf{c}' \in \mathcal{C}_t^\mathbb{L}$, we have $B_t^\mathbb{L} (\mathbf{c}) \cap B_t^\mathbb{L} (\mathbf{c}') = \emptyset$.
    The maximum cardinality of a $t$-strand-loss-correcting code is denoted by $\maxcar{L}{t}$.
\end{definition}

\begin{problem}\label{pro:strandloss}
 Find the value of $\maxcar{L}{t}$ and $t$-strand-loss-correcting codes of cardinality $\maxcar{L}{t}$.
\end{problem}

\begin{definition}\label{def:deletion}
    Let $\mathbf{x} \in [0,M]^n$ be a composite vector with a possible matrix representation $X \in \{ 0,1 \}^{M \times n}$. $\mathbb{D}$ denotes the deletion error type. It is said that $t$ \textbf{deletions} occurred in $X$, if for $t$ tuples $(k_{\ell}, h_{\ell}), k_\ell \in [1,M], h_\ell \in [1,n], \ell \in [1, t]$, deleting the elements $X_{k_{\ell},h_{\ell}}$, $\forall\ell \in [1, t]$, from $X$ and shifting the respective row to the left results in $R$, which has $M$ rows and each row has length at most~$n$. 
    A code $\mathcal{C}_t^\mathbb{D} \subseteq [0,M]^n$ is called a \textbf{$t$-deletion-correcting code} if for any $\mathbf{c}, \mathbf{c}' \in \mathcal{C}_t^\mathbb{D}$, we have $B_t^\mathbb{D} (\mathbf{c}) \cap B_t^\mathbb{D} (\mathbf{c}') = \emptyset$. 
    The maximum cardinality of a $t$-deletion correcting code is denoted by $\maxcar{D}{t}$.
\end{definition}

\begin{problem}\label{pro:deletion}
Find the value of $\maxcar{D}{t}$ and deletion-correcting codes of cardinality $\maxcar{D}{t}$. 
\end{problem}

The definitions for insertion ($\mathbb{I}$) and indel ($\mathbb{ID}$) errors are analogue to \autoref{def:deletion}. For completeness, the specific definitions can be found in the appendix. 
\newcommand{\defInsertion}{
\begin{definition}\label{def:insertion}
    Let $\mathbf{x} \in [0,M]^n$ be a composite vector with a possible matrix representation $X \in \{ 0, 1\}^{M \times n}$. $\mathbb{I}$ denotes the insertion error type. 
    It is said that $t$ \textbf{insertions} occurred, if there exist $t$ tuples $(h_\ell, k_\ell, s_\ell), q \in [0,M], h_\ell \in [0,n], s_\ell \in \{0,1\}, \ell \in [1, \ldots , t]$ which denote the indices where the symbol $s_\ell$ is inserted in $X_{h_\ell,k_\ell}$. Consecutive positions are shifted to the right such that each row of $R$ has length at most $n+t$.
\end{definition}
}

\newcommand{\defIndel}{
\begin{definition}\label{def:indel}
    Let $\mathbf{x} \in [0,M]^n$ be a composite vector with a possible matrix representation $X \in \{ 0, 1\}^{M \times n}$. The $t$-insertion-deletion (indel) error type is denoted by $\mathbb{I}\mathbb{D}$. It is said that $t$ \textbf{indel errors} occurred if $t_i$ insertions occurred together with $t_d$ deletions occurred and $t = t_d + t_i$. 
    A code $\mathcal{C}_t^\mathbb{ID} \subseteq [0,M]^n$ is called a \textbf{$t$-indel-correcting code} if for any $\mathbf{c}, \mathbf{c}' \in \mathcal{C}_t^\mathbb{ID}$, we have $B_t^\mathbb{ID} (\mathbf{c}) \cap B_t^\mathbb{ID} (\mathbf{c}') = \emptyset$. 
    The maximum cardinality of a $t$-indel correcting code is denoted by $\maxcar{ID}{t}$.
\end{definition}
}

\begin{problem}\label{pro:indel}
Find the value of $\maxcar{ID}{t}$ and indel-correcting codes of cardinality $\maxcar{ID}{t}$. 
\end{problem}
\begin{example}\label{exa:errors}
    Let the composite vector be $\mathbf{x} = ( 3, 5, 3, 2)$. Then, one possible matrix representation $X$ and possible received matrices $R^\mathbb{E}$ with errors of type $\mathbb{E}$ are given below. The red symbols or lines indicate where the error occurred. \\[-2ex]
    {\footnotesize
    \begin{align*}
        X &= 
        \left[ \begin{array}{cccc}
          0 & 1 & 1 & 0 \\
          1 & 1 & 0 & 0 \\
          0 & 1 & 1 & 0 \\
          1 & 1 & 1 & 1 \\
          1 & 1 & 0 & 1 
        \end{array} \right] 
        \\
        R^\mathbb{S} &= 
        \left[ \begin{array}{cccc}
          0 & 1 & 1 & 0 \\
          1 & 1 & {\color{red} 1} & 0 \\
          0 & 1 & 1 & 0 \\
          1 & 1 & 1 & 1 \\
          1 & 1 & 0 & 1 
        \end{array} \right]
        \quad
        R^\mathbb{L} = 
        \left[ \begin{array}{cccc}
          0 & 1 & 1 & 0 \\
          1 & 1 & 1 & 0 \\
          \arrayrulecolor{red} \hline
          1 & 1 & 1 & 1 \\
          1 & 1 & 0 & 1 
        \end{array} \right]
        \\
        R^\mathbb{D} &= 
        \left[ \begin{array}{cccc}
          0 & 1 & 1 & 0 \\
          1 & \color{red} \vline \color{black} 0 & 0 & \\
          0 & 1 & 1 & 0 \\
          1 & 1 & 1 & 1 \\
          1 & 1 & 0 & 1 
        \end{array} \right]
        \quad
        R^\mathbb{I} = 
        \left[ \begin{array}{ccccc}
          0 & 1 & 1 & 0 & \\
          1 & 1 & 0 & 0 & \\
          0 & 1 & {\color{red} 0 } & 1 & 0 \\
          1 & 1 & 1 & 1 & \\
          1 & 1 & 0 & 1 & 
        \end{array} \right]
    \end{align*}
    }
\end{example}

\section{Substitution Errors} \label{sec:substitution}
To analyze $t$-substitution-correcting codes, we define codes in the $L_1$-metric (also known as the Manhattan distance) in \autoref{def:l1-code} and show their equivalence to $t$-substitution-correcting codes in \autoref{cla:l1-equiv}. 

\begin{definition}\label{def:l1-code}
     For two vectors $\mathbf{x}, \mathbf{y} \in [0,M]^n$, the $L_1$-distance $d_1(\mathbf{x}, \mathbf{y})$ is defined as 
     $ d_1(\mathbf{x}, \mathbf{y})= \sum_{j=1}^{n} \left| x_j - y_j \right|$.
     A code $\mathcal{C} \subseteq [0,M]^n$ has minimum $L_1$-distance $d$, if for all $\mathbf{x}, \mathbf{y} \in \mathcal{C}$ we have 
    $d_1(\mathbf{x}, \mathbf{y}) \geq d $.
    The maximum cardinality of a code of length $n$ over an alphabet of size $q$ with minimum $L_1$-distance $d$ is denoted as $A_n^{L_1} (q, d)$.
\end{definition}

\begin{claim}\label{cla:l1-equiv}
    For all vectors $\mathbf{x}, \mathbf{y} \in [0,M]^n$ we have the equivalence
    \begin{align*}
        d_1(\mathbf{x}, \mathbf{y})  \geq 2t+1  \iff B_t^\mathbb{S}(\mathbf{x}) \cap B_t^\mathbb{S}(\mathbf{y}) = \emptyset.
    \end{align*}
\end{claim}

\newcommand{\proClaSubstitution}{
\begin{proof}
    $(\Rightarrow )$: Assume that $\mathbf{x}, \mathbf{y} \in [0,M]^n $ are two words with $d_1(\mathbf{x}, \mathbf{y}) \geq 2t +1$ and possible matrix representations $X, Y$ such that $\exists R \in B_t^\mathbb{S}(\mathbf{x}) \cap B_t^\mathbb{S}(\mathbf{y}) $. 
    Therefore, we introduce the binary error matrices $E, E'$ such that $R = X + E$ and $R = Y + E'$. $E$ and $E'$ can have at most $t$ nonzero entries as each entry represents a substitution.
    Then we obtain
    \begin{align*}
        d_1\left( \mathbf{x}, \mathbf{y} \right) 
        &= \sum_{j=1}^n \left| x_j - y_j' \right| 
        = \sum_{j=1}^n \left| \sum_{i = 1}^M X_{i,j} - Y_{i,j} \right| 
        \\
        &= \sum_{j=1}^n \left| \sum_{i = 1}^M R_{i,j} - E_{i,j} - R_{i,j} + E_{i,j}' \right| 
        \\
        &= \sum_{j=1}^n \left| \sum_{i = 1}^M E'_{i,j} - E_{i,j} \right| 
        \\
        & \leq \sum_{j=1}^n \sum_{i = 1}^M \left| E_{i,j} \right| + \sum_{j=1}^n  \sum_{i = 1}^M \left| E_{i,j}' \right| 
        \leq 2t 
    \end{align*}
    This contradicts the initial assumption that $d_1(\mathbf{x}, \mathbf{y}) \geq 2t +1$ and therefore, $B_t^\mathbb{S}(\mathbf{x}) \cap B_t^\mathbb{S}(\mathbf{y}) = \emptyset $. 
    \\
    $(\Leftarrow) $: 
    Now, assume that $\mathbf{x}, \mathbf{y} \in [0,M]^n $ are chosen such that $B_t^\mathbb{S}(\mathbf{x}) \cap B_t^\mathbb{S}(\mathbf{y}) = \emptyset$, where $\mathbf{x}, \mathbf{y}$ have distance $d_1(\mathbf{x}, \mathbf{y}) < 2t +1$. The matrices $X, Y$ are matrix representations of $\mathbf{x}, \mathbf{y}$, structured such that in each column, all zeros are on the top and all ones are on the bottom. Then, for each column $j$, there exist two integers $e_j^{(0)} \leq e_j^{(1)} \in [1,M]$ such that 
    \begin{align*}
        X_{i,j} = Y_{i,j} & \quad \text{if } j<e_j^{(0)} \text{ or } j \geq e_j^{(1)}
        \\
        X_{i,j} \neq Y_{i,j} & \quad \text{otherwise}
    \end{align*}
    With this structure of $X, Y$, we get that $e_j^{(1)} - e_j^{(0)} = |x_j - y_j'| $ and $X, Y$ can only differ in $\sum_{j=1}^n \left(e_j^{(1)} - e_j^{(0)}\right) = d_1(\mathbf{x}, \mathbf{y}) < 2t + 1$ positions. Next, define a matrix $E$ such that it has ones in $t$ positions, where $X, Y$ differ and $E'$ in the remaining up to $t$ positions such that we get $X + E + E' = Y$. Finally, we define a matrix $R=X+E=Y+E'$ and since $E, E' $ introduce at most $t$ substitutions we get that $B_t^\mathbb{S}(\mathbf{c}) \cap B_t^\mathbb{S}(\mathbf{c}') \neq \emptyset $ which contradicts the assumption.  
\end{proof}
}

As a result of \autoref{cla:l1-equiv}, a code $\mathcal{C}_t^\mathbb{S}$ is a $t$-substitution-correcting code if and only if its minimum $L_1$-distance is at least $2t+1$. 
Hence, the following equality holds: 
\begin{equation*}
    \maxcar{S}{t} = A_n^{L_1} (M+1, 2t +1).
\end{equation*}

To the best of the authors' knowledge, codes in this metric are only little studied. A variation of the $L_1$-distance which considers whether the errors increase or decrease the levels was studied in \autocite{tallini_2011_L1distanceError, tallini_2012_SymmetricL1, chen_2021_OptimalCodes}. In \autocite{etzion_2013_CodingLee}, the $L_1$-distance was studied but over the infinite alphabet of all integers and anticodes over this metric were studied in \autocite[Chapter 2]{ahlswede_2008_LecturesAdvances}. Thus, we are not aware of explicit results on the value of $A_n^{L_1} (M+1, 2t +1)$, besides some trivial and special cases. 

\section{Loss of Strands}\label{sec:strandloss}

This section discusses error events in which $t$ of the strands are lost. 
Similar to \autoref{sec:substitution}, we first prove the equivalence to codes in the $L_\infty$-metric.

\begin{definition}\label{def:linfty-vector}
     For two vectors $\mathbf{x}, \mathbf{y} \in [0,M]^n$, the $L_\infty$-distance $d_\infty(\mathbf{x}, \mathbf{y})$ is defined as 
     $d_\infty(\mathbf{x}, \mathbf{y}) = \max_{j \in [1,n]} \left| x_j - y_j \right|$.
     A code $\mathcal{C} \subseteq [0,M]^n$ has minimum $L_\infty$-distance $d$, if for all $\mathbf{x}, \mathbf{y} \in \mathcal{C}$ we have 
    $ d_\infty(\mathbf{x}, \mathbf{y}) \geq d $.
    The maximum cardinality of a code of length $n$ over an alphabet of size $q$ with minimum $L_\infty$-distance $d$ is denoted as $A_n^{L_\infty} (q, d)$.
\end{definition}

\begin{claim}\label{cla:linfty}
    For any two vectors $\mathbf{x}, \mathbf{y} \in [0,M]^n$ we have that
    \begin{align*}
        d_\infty(\mathbf{x}, \mathbf{y})  \geq t+1  \iff B_t^\mathbb{L}(\mathbf{x}) \cap B_t^\mathbb{L}(\mathbf{y}) = \emptyset.
    \end{align*}
\end{claim}
\newcommand{\proClaStrandLoss}{
\begin{proof}
    $(\Rightarrow )$: Let $\mathbf{x}, \mathbf{y} \in [0,M]^n $ be two words with $d_\infty(\mathbf{x}, \mathbf{y}) \geq t +1$ and possible matrix representations $X, Y$. Assume to the contrary that $\exists R \in B_t^\mathbb{L}(\mathbf{x}) \cap B_t^\mathbb{L}(\mathbf{y}) $. 
    Therefore, there exist binary matrices $E, E' \in \{0,1 \}^{t \times n}$
    \begin{align*}
        X &= 
        \begin{bmatrix}
            E \\ R
        \end{bmatrix}
        \quad \text{and} \quad
        Y = 
        \begin{bmatrix}
            E' \\ R
        \end{bmatrix}
        \\
        \text{such that } x_j &= \sum_{i = 1}^M X_{i,j} = \sum_{i = 1}^t E_{i,j} + \sum_{i = 1}^{M-t} R_{i,j} 
        \\
        \text{and } y_j &= \sum_{i = 1}^M Y_{i,j} = \sum_{i = 1}^t E_{i,j}' + \sum_{i = 1}^{M-t} R_{i,j} 
    \end{align*}
    up to some permutation in $X, Y$. 
    Then we get 
    \begin{align*}
        d_\infty&(\mathbf{x}, \mathbf{y}) 
        = \max_{j \in [1,n]} \left| x_j - y_j \right|
        \\
        &= \max_{j \in [1,n]} \left| \sum_{i = 1}^t E_{i,j} + \sum_{i = 1}^{M-t} R_{i,j} - \sum_{i = 1}^t E_{i,j}' + \sum_{i = 1}^{M-t} R_{i,j} \right|
        \\
        &= \max_{j \in [1,n]} \left| \sum_{i = 1}^t E_{i,j} - \sum_{i = 1}^t E_{i,j}' \right| \leq t
    \end{align*}
    This contradicts the initial assumption that $d_\infty(\mathbf{x}, \mathbf{y}) \geq t +1$ and therefore, $B_t^\mathbb{L}(\mathbf{x}) \cap B_t^\mathbb{L}(\mathbf{y}) = \emptyset $. 
    \\
    $(\Leftarrow) $: 
    Now, let $\mathbf{x}, \mathbf{y} \in [0,M]^n $ be two words, such that $B_t^\mathbb{L}(\mathbf{x}) \cap B_t^\mathbb{L}(\mathbf{y}) = \emptyset$. Assume to the contrary  $d_\infty(\mathbf{x}, \mathbf{y}) < t +1$. 
    Let us assume $X$ and $Y$ are some matrix representations of $\mathbf{x}$ and  $\mathbf{y}$, respectively. 
    For each $1 \le j \le n$, the $j$-th entry of $\mathbf{x}, \mathbf{y}$, which are denoted by $x_j, y_j$ satisfy the following. If $x_j \geq y_j$, then there are $y_j$ $1$'s and $M-x_j$ $0$'s in the $j$-th column of both $X, Y$. Furthermore, we have $y_j + M - x_j \geq M - t$. If  $x_j < y_j$, there are $x_j$ $1$'s and $M-y_j$ $0$'s in column $j$ of both $X, Y$ and we have $x_j + M - y_j \geq M - t$. This implies that for each column in $X$ and $Y$ there are at least $M-t$ equal elements. Therefore, we can construct $R \in \{0,1 \}^{(M-t) \times n}$ which is a submatrix of both $X$ and $Y$. Thus, $ R \in B_t^\mathbb{L}(\mathbf{x}) \cap B_t^\mathbb{L}(\mathbf{y})$ which contradicts the initial assumption and the claim is proven.  
\end{proof}
}

\subsection{Bounds on the Size of Codes for Correcting Loss of Strands} 

First, we will introduce the following general proposition about the size of codes in partitions. 

\begin{proposition}\label{pro:partition}
    Let $\maxcar{E}{t} $ be the maximum cardinatliy of a code able to correct $t$ errors of type $\mathbb{E}$ in $[0,M]^n$. Furthermore, let $\mathcal{P}_1 ,\dots , \mathcal{P}_r$ for a positive integer $r \in \mathbb{N}$ be an exhaustive partition of $[0,M]^n $ such that $\bigcup_{i \in [1,r]} \mathcal{P}_i = [0,M]^n $. Let us denote by  $A_n^{\mathbb{E},1} \! \left(M, t \right) , \ldots , A_n^{\mathbb{E}, r} \! \left(M, t \right) $ the maximal cardinality of codes in each partition, which are able to correct $t$ errors of type $\mathbb{E}$. Then, we get 
    \begin{align*}
         \maxcar{E}{t} \leq \sum_{i=1}^r A_n^{\mathbb{E},i} \! \left(M, t \right) .
    \end{align*}
\end{proposition}

\begin{proof}
    Assume $\mathcal{C}_t^\mathbb{E} $ is a code of maximum size in $[0,M]^n$ and for $1\leq i\leq r$, let $\mathcal{C}_t^{\mathbb{E},i} =  \mathcal{C}_t^\mathbb{E} \cap \mathcal{P}_i$. 
    Since the partition is exhaustive, we get that $\bigcup_{i=1}^r \mathcal{C}_t^{\mathbb{E},i} = \mathcal{C}_t^\mathbb{E}$.
    Then, it holds that $\mathcal{C}_t^{\mathbb{E},i}$ can correct $t$ errors of type $\mathbb{E}$ in $\mathcal{P}_i$, which assures that $|\mathcal{C}_t^{\mathbb{E},i}|\leq A_n^{\mathbb{E},i} \! \left(M, t \right) $. Therefore, $\maxcar{E}{t} = |\mathcal{C}_t^\mathbb{E}| \leq \sum_{i=1}^r|\mathcal{C}_t^{\mathbb{E},i}|\leq \sum_{i=1}^r A_n^{\mathbb{E},i} \! \left(M, t \right)$.
\end{proof}
Note that the inequality follows if the partition is not disjoint or if there exist codewords in one set of the partition which are confusable with codewords in other sets of the partition. 
We can now use this proposition to design suitable partitions of the set of composite vectors and derive the following upper bound.
\begin{theorem}\label{the:strandloss-partition}
    The maximum cardinality of a $t$-strand-loss-correcting code is given by 
    \begin{align*}
        \maxcar{L}{t} = A_n^{L_\infty} (M+1, t+1)\leq \left\lceil \frac{M+1}{t+1} \right\rceil^n .
    \end{align*}
\end{theorem}
\begin{proof} 
    Consider a partition of $[0,M]^n$ with the sets 
    \begin{align*}
        \mathcal{P}_\mathbf{u} \! = [u_1,u_1 \! +t] \times [u_2,u_2 \! + t] \times \! \cdots \! \times [u_n, u_n  \!+ t] \cap [0,M]^n,
        \\
        \forall \mathbf{u} 
        \in \left\{0, (t+1), 2(t+1),  \ldots, \left\lfloor \frac{M}{t+1}\right\rfloor  (t+1)\right\}^n.
    \end{align*}
    Thus, each $u_j$ for $j \in [1,n]$ is a multiple of $t+1$ and the sets $\mathcal{P}_\mathbf{u}$ are mutually disjoint and form a partition of $[0,M]^n$. 
    Furthermore, for every $\mathbf{p}, \mathbf{p}' \in \mathcal{P}_\mathbf{u}$ it holds that $d_\infty(\mathbf{p}, \mathbf{p}') < t+1 $, and thus the largest size of a code on every partition is 1. There are $\left\lfloor \frac{M}{t+1} +1 \right\rfloor^n = \left\lceil \frac{M+1}{t+1} \right\rceil^n$ sets $\mathcal{P}_\mathbf{u}$ which form an exhaustive partition of $[0,M]^n$. Together with \autoref{pro:partition}, we get that $ A_n^{L_\infty} (M+1, t+1) \leq \left\lceil \frac{M+1}{t+1} \right\rceil^n$ and with \autoref{cla:linfty} the theorem is proven.
\end{proof}

\subsection{Code Construction for Strand Loss Errors}

\begin{construction}\label{con:strandloss}
    Let $n,M>t$ be positive integers and let
    \begin{align*}
        \mathcal{C}_t^\mathbb{L} = \left\{ \mathbf{c} \in [0,M]^n : c_j \equiv 0 \mod{t +1} , \forall j \in [0,n] \right\} .
    \end{align*}
\end{construction}

\begin{theorem}\label{the:strand-loss-code}
    The code $\mathcal{C}_t^\mathbb{L}$ from \autoref{con:strandloss} is a $t$-strand-loss-correcting code.
\end{theorem}
\newcommand{\proTheStrandLossCode}{
\begin{proof}
    Consider the matrix representations $C, C'$ of $\mathbf{c}, \mathbf{c}' \in \mathcal{C}_t^\mathbb{L}$ and their resulting matrices $R, R'$ after $t$ strand losses. The number of ones in each column of $C, C'$ must be a multiple of $t+1$, say $k (t+1)$. Through $t$ strand losses, the number of ones can only decrease by $t$. Therefore, we get 
    \begin{align*}
        (k-1) (t+1) < \sum_{i=1}^M R_{i, j} \leq k(t+1) \forall j \in [1,n]
    \end{align*}
    Thus, the number of $1$'s in column $j$ of $R$ and $R'$ can only be equal if $c_j = c_j'$ for all $j \in [1,n]$. Hence, if $c_j \neq c_j'$ for any $j$, the resulting matrices must be different, which concludes the proof. 
\end{proof}
}

The cardinality of the code is given by 
\begin{align*}
    \left| \mathcal{C}_t^\mathbb{L} \right| = \left\lceil \frac{M+1}{t+1} \right\rceil^n = \maxcar{L}{t},
\end{align*}
which meets the bound of \autoref{the:strandloss-partition}. Therefore, we see that the bound is tight and the code optimal. 

\section{Deletion Errors} \label{sec:deletion}

In the following, we will analyze $t$-deletion correcting codes and solve \autoref{pro:deletion}. 

\begin{claim}\label{cla:indel}
     For $\mathbf{c}, \mathbf{c}' \in [0,M]^n$, we have the equivalence
     \begin{align*}
         B_t^\mathbb{ID}(\mathbf{c}) \cap B_t^\mathbb{ID}(\mathbf{c}') = \emptyset \iff B_t^\mathbb{D}(\mathbf{c}) \cap B_t^\mathbb{D}(\mathbf{c}') = \emptyset.
     \end{align*}
\end{claim}
\begin{proof}
    We will show the proof for insertion errors. The same argument holds for indel errors. 
    Let $\mathbf{c}, \mathbf{c}'  \in [0,M]^n$ with corresponding matrix representation $C, C'$ such that $B_t^\mathbb{D}(\mathbf{c}) \cap B_t^\mathbb{D}(\mathbf{c}') = \emptyset$ and assume in the contrary that there exists an $R^\mathbb{I} \in B_t^\mathbb{I}(\mathbf{c}) \cap B_t^\mathbb{I}(\mathbf{c}')$. 
    $R^\mathbb{I}$ can be received by inserting $t$ symbols at position $(k_\ell, h_\ell), \ell \in [1,t]$ in $C$ and $t$ symbols at $(k_\ell', h_\ell'), \ell \in [1,t]$ in $C'$. 
    As the length of the rows change, the insertions must happen in the same rows so for each $\ell \in [1,t]$ we have one $k_\ell = k_\ell'$. For two binary vectors of length $n$ with $d<n$, it is known from \autocite{levenshtein_1966_BinaryCodes} that they share an element in the $d$-deletion ball if and only if they share an element in the $d$-insertion ball. This result can be applied to every row of $C, C'$ respectively $R$ affected by insertions and we get that there exists an $R^\mathbb{D} \in B_t^{\mathbb{D}}(\mathbf{c}) \cap B_t^{\mathbb{D}}(\mathbf{c}')$, which contradicts the assumption. The converse follows with the same argument. 
\end{proof}

Using the equivalence of \autoref{cla:indel}, if we solve \autoref{pro:deletion}, then \autoref{pro:indel} is also already covered. 

\subsection{Size of Error Balls for Single Deletion Errors}

To derive the error ball size, we use the following definitions. Let $\run{\mathbf{y}}$ denote the number of runs in the binary vector $\mathbf{y}$ and $V(\mathbf{x})$ be the set of all binary vectors which can be a row in $X \in \setX$,
    \vspace{-0.2cm}
    \begin{align*}
        V(\mathbf{x}) = \{ \mathbf{y} \in \{0,1 \}^n : \nexists j \in [1,n] : x_j = M \text{ and } y_j = 0 \\ \text{ or } x_j = 0 \text{ and } y_j = 1 \}.
    \end{align*}

\begin{theorem}
    The error ball size for a single deletion is given by 
        \vspace{-0.2cm}
    \begin{equation*}
        |B_1^\mathbb{D}(\mathbf{x})| = M \sum_{\mathbf{y} \in V(\mathbf{x})}  \run{\mathbf{y}} \prod_{j=1}^{n} \binom{M-1}{x_j - y_j}.
    \end{equation*}
\end{theorem}
\begin{proof}
    The proof follows by considering all possible deletions. First, we let $\textbf{y}\in V(\textbf{x)}$ be a  row in the matrix representation of the vector $\textbf{x}$, and assume a deletion error occurred in $\textbf{y}$. Note, that the vector $\textbf{y}$ can be located in any of the $M$ rows of the matrix representation $X$, whereas each of these locations results  a different matrix $X$. Furthermore, in~\autocite{levenshtein_1966_BinaryCodes}, it was shown that the number of words that can be obtained by deleting a symbol in  $\textbf{y}$ is given by $\run{\mathbf{y}}$. Additionally, if the location of $\textbf{y}$  is already selected, the remaining rows of $X$ consist of all possible $X' \in \mathcal{X}\!\left( \mathbf{x} - \mathbf{y} \right) \subseteq \{0,1 \}^{(M-1) \times n}$. The number of such $X'$ is given by $\prod_{j=1}^{n} \binom{M-1}{x_j - y_j}$. 
    Finally, note that even though two different vectors $\mathbf{y}$ could end up as the same vector after the deletion occurred, the remaining rows $X'$ will be different, and thus the resulting element of $B_1^\mathbb{D}(\mathbf{x})$ will be different for every choice of $\mathbf{y} \in V(\mathbf{x})$. 
\end{proof}

\subsection{Upper Bound on the Size of Deletion-Correcting Codes}

To derive an upper bound on $t$-deletion-correcting codes as in Definition~\ref{def:deletion}, we will reduce them to classical binary deletion-correcting codes. Therefore, in the remainder of the section, we restrict $t\leq n$. Further notation about these codes is similar to \autocite{levenshtein_1966_BinaryCodes, sloane_2002_SingleDeletionCorrectingCodes}.

\begin{definition}
    Let $\mathbf{x} \in \{ 0, 1 \}^n$ be a binary vector of length~$n$. It is said that \textbf{$t$ deletion errors} occurred in $\mathbf{x}$, if there exist $t$ positions $h_\ell \in [1,n], \ell \in [1,t] $, which are removed from $\mathbf{x}$ to obtain $\mathbf{r} \in \{ 0, 1\}^{n- t}$. In this setup, a code that can correct $t$ deletion errors is called a $t$-deletion correcting code over binary vectors. The largest code cardinality of such code is denoted by $D(n,t)$. 
\end{definition} 

\begin{theorem}\label{the:deletion_cardinality}
    For $t < n$, the cardinality of a $t$-deletion-correcting code $\maxcar{D}{t}$ is bounded from above by 
    \begin{equation*}
        \maxcar{D}{t} \leq \left( \left\lceil \frac{M+1}{2} \right\rceil \right)^n \cdot D(n,t).
    \end{equation*}
\end{theorem}
\begin{proof}
    We distinguish between two cases. 
    \\
    \emph{Case 1 ($M$ is odd):} 
    Consider a partition of $[0,M]^n$ with the sets 
    \begin{equation*}
        \mathcal{P}_\mathbf{u} = [u_1, u_1 + 1] \times [u_2, u_2 + 1] \times \dots \times [u_n, u_n + 1],
    \end{equation*}
    for all $ \mathbf{u} = (u_1, \ldots , u_n ) \in \{0, 2, \dots, M-1\}^n$. Let $\mathbf{p} \in \mathcal{P}_\mathbf{u}$ be a composite vector. 
    Next, we define the vector $\mathbf{u}_\mathbf{p} = \mathbf{p} \mod 2 \in \{0,1\}^n$ such that we get $\mathbf{p} = \mathbf{u} + \mathbf{u}_\mathbf{p}$. Then, we can obtain a matrix representation 
    \begin{equation}
        P = 
        \begin{bmatrix}
            A_\mathbf{p} \\ \mathbf{u}_\mathbf{p}
        \end{bmatrix}
        \label{eq:p-matrix}
    \end{equation}
    with a submatrix $A_\mathbf{p}$ which is the same for all $\mathbf{p} \in \mathcal{P}_\mathbf{u}$. Thus, for $\mathbf{p}, \mathbf{p}' \in \mathcal{P}_\mathbf{u}$ with corresponding matrix representation $P, P'$ as in \eqref{eq:p-matrix}, we get that $P$ and $P'$ are not confusable under $t$ deletions, only if the binary vectors $\mathbf{u}_\mathbf{p}$ and $\mathbf{u}_{\mathbf{p}'}$ are not confusable under $t$ deletions. 
    Therefore, a $t$-deletion-correcting code $\mathcal{C}_t^\mathbb{D}$ on each set $\mathcal{P}_\mathbf{u}$ has cardinality at most $D(n,t)$. There are $\left( (M+1)/2 \right)^n$ possible sets $\mathcal{P}_\mathbf{u}$ which form an exhaustive partition of $[0,M]^n$. Together with \autoref{pro:partition}, the theorem follows.
    \\
   \emph{Case 2 ($M$ is even):} 
    We choose $\mathcal{P}_\mathbf{u}$ such that $ \mathbf{u} = (u_1, \ldots , u_n ) \in \{0,1, 3, 5,  \ldots, M-1\}^n$ and follow the same argument as in Case 1.
\end{proof}

\subsection{Construction of a Single-Deletion-Correcting Code}

In the following, we present a single-deletion-correcting code (SDC). Our proposed construction is based on the Varshamov-Tenengolts (VT) codes \autocite{sloane_2002_SingleDeletionCorrectingCodes,varshamov_1965_CodesWhich}, which were proven by Levenshtein~\cite{levenshtein_1966_BinaryCodes} to be binary single-strand single-deletion correcting codes. 

The \emph{$VT$ syndrome} of a vector $\textbf{x} =(x_1, \ldots, x_n) \in\{ 0,1\}^n$, denoted by $s(\textbf{x})\in [0,n]$, is defined as $ s(\textbf{x}) \triangleq \sum_{j=1}^n j x_j \mod(n+1)$.   For $a\in[0,n]$, the length-$n$ VT code with parameter $a$, denoted by $VT_a(n)$, is defined as follows. 
\begin{equation*}
VT_a(n) \triangleq \left\{ \mathbf{x} \in \{0,1\}^n : s(\textbf{x})= a \right\}.
\end{equation*}

It turns out that a small adaption is sufficient to apply these codes to correct a single deletion in our composite channel model, as seen in the next construction. 

\begin{construction}
Let $M\ge1$ be a positive integer and let $a \in [0,n]$.
    \vspace{-0.2cm}
    \begin{equation*}
        \mathcal{C}_1^{\mathbb{D}}(a) = \Big\{
        \mathbf{c} \in [0,M]^n : \sum_{j=1}^{n} j \cdot c_j \equiv a \mod (n+1)
        \Big\}
    \end{equation*}
\end{construction}

\begin{theorem}\label{the:sdc}
The code $\mathcal{C}_1^{\mathbb{D}}(a)$ is an SDC for all $a \in [0,n]$. 
\end{theorem}
\begin{proof}
We prove the statement for $a=0$, and we denote $\mathcal{C}_1^{\mathbb{D}}(a) \triangleq \mathcal{C}_1^{\mathbb{D}}$. The proof is the same for other values of $a$.  Let $\textbf{c} = (c_1, \ldots, c_n)\in \mathcal{C}_1^\mathbb{D}$ and let $C \in \{0,1\}^{M\times n}$ be a possible matrix representation of $\textbf{c}$, where for $1\le i \le M$, $C_i$ denotes the $i$-th row of $C$. By the code definition, we have that,  
 \begin{align*}
 \sum_{i=1}^M s(C_i) &= \sum_{i=1}^M \sum_{j=1}^n j C_{i, j}  =  \sum_{j=1}^n  \sum_{i=1}^M j C_{i, j} \\&= \sum_{j=1}^n j \sum_{i=1}^M  C_{i, j} = \sum_{j=1}^n j c_{j} \equiv 0 \mod (n+1).
 \end{align*}

 Next, let us assume a deletion occurred in the $k$-th row of $C$, for some $1\le k \le M$.  
 Since the rest of the rows of $C$ have not experienced any error event, it is possible to calculate their syndromes. Therefore, it holds that, 
 \begin{align*}
 \sum_{i=1, i \ne k}^M s(C_i) \equiv 0-s(C_k) \mod (n+1),
\end{align*}
which implies that the VT syndrome of $C_k$ can be retrieved and denote its value by $s(C_k) = b \in [0,n]$. Note that in this case the $k$-th row of $C$ is a codeword in the code $VT_b(n)$ that can correct a single deletion. Thus, it is possible to correct the deleted symbol in the $k$-th row of $C$ by using the decoder of the code $VT_b(n)$.
\end{proof}

\begin{corollary}
    There exists an $a \in [0,n]$ such that
    \begin{align*}
       \maxcar{D}{1}\ge \left| \mathcal{C}_1^{\mathbb{D}}(a) \right| \ge \left\lceil \frac{(M+1)^n}{n+1} \right\rceil .
    \end{align*}
\end{corollary}

Now recall that Theorem~\ref{the:deletion_cardinality} states an upper bound on the size of SDC, given by $\maxcar{D}{t} \leq \left( \left\lceil \frac{M+1}{2} \right\rceil \right)^n \cdot D(n,t)$, while in \autocite{kulkarni_2013_NonasymptoticUpper} it was shown that under binary alphabet $D(n,1) \le \frac{2^n-2}{n-1}.$ Thus, combining these two results for odd $M$, we get that,  $\maxcar{D}{1} \leq \left(\frac{M+1}{2} \right)^n \cdot\frac{2^n-2}{n-1}\leq \frac{(M+1)^n}{n-1}$. The latter implies that the construction of $\mathcal{C}_1^{\mathbb{D}}(a)$ is asymptomatically optimal for odd $M$. 

\begin{example}\label{exa:del-correct}
  If we look at \autoref{exa:errors}, we have $\mathbf{x} = (3, 5, 3, 2) $ and the $s(\mathbf{x}) \mod 5 = 0 $. 
  The shorter strand in which the deletion must have occurred is the second row of $R^\mathbb{D}$ and the VT-code syndrome of the remaining rows is given by $s(2, 4, 3, 2) \mod 5 = 2 $. This means that we can correct the row by treating it as a binary VT-code with syndrome $3$ using the same algorithm as in \autocite{levenshtein_1966_BinaryCodes,sloane_2002_SingleDeletionCorrectingCodes}.  
\end{example}

\section{Combination of Error Types}
In this section we will consider a combination of $t$ strand loss errors and one substitution error. 
\begin{construction}
    Let $M>1$ and $t<M$ positive integers. Further, let $\mathcal{C}_H$ be a binary length-$n$ single substitution correcting code (e.g. Hamming code). 
    \vspace{-0.2cm}
    \begin{align*}
        \mathcal{C}_{t,1}^{\mathbb{L,S}} \! \! = \! \left\{ \!
            \mathbf{c} \in [0,M]^n \!  : c_j \equiv 0 \! \! \! \! \mod{t+1}, \frac{\mathbf{c}}{t+1} \! \! \! \!  \mod{2} \in \mathcal{C}_H
        \right\}
    \end{align*}
\end{construction}
\begin{theorem}
    The code $\mathcal{C}_{t,1}^\mathbb{L,S}$ can correct $t$ strand loss errors and one substitution error. 
\end{theorem}
\begin{proof}
    We will prove the correctness of this construction by providing an explicit decoding algorithm. 
    Let $\mathbf{c} \in \mathcal{C}_{t,1}^\mathbb{L,S}$ and $C$ its corresponding matrix. Let $R$ be the erroneous version of $C$, and assume it has experienced $t$ strand losses and one substitution. By summing each column of $R$, we get the vector $\mathbf{r}$. Assume that the symbol $s$ was substituted in column $h$. First, we correct the strand loss errors and ignore the substitution. Since we know from \autoref{sec:strandloss} that $c_j - t \leq r_j \leq c_j$, we can receive the vector $\mathbf{y}$ such that $y_j = r_j + (- r_j \! \! \mod{t+1})$. This would be the codeword $\mathbf{c}$ if no substitution would have appeared. Therefore, we notice that $y_j = c_j$ for all $j \neq h$.
    \\
    Next, let us consider the column $h$. We have $r_h \in \left\{ c_h - t - 1, c_h - t, \ldots , c_h , c_h +1  \right\}$. If $r_h \in \left\{ c_h - t, \ldots , c_h \right\}$, we will get $y_h = c_h$ and there is no error when applying the decoder of the Hamming code. Hence, we only need to consider the other two cases. 
    \\
    Case 1 ($r_h = y_h = c_h - t - 1$):
    After decoding the strand losses, we receive $y_h = c_h - t - 1$ and furthermore $\frac{y_h}{t+1} = \frac{c_h}{t+1} - 1$. Thus, the Hamming decoder will recognize an error at position $h$.  
    Since we did not change position $h$ while correcting the strand losses ($y_h = r_h$), we know that $c_h = r_h + t +1$. 
    \\
    Case 2 ($r_h = y_h - t = c_h +1 $):
    After decoding the strand losses, we receive $y_h = c_h + t + 1$ and furthermore $\frac{y_h}{t+1} = \frac{c_h}{t+1} + 1$ Thus, again the Hamming decoder will recognize an error at position $h$. 
    Since we did change position $h$ while correcting the strand losses ($y_h = r_h + t$), we know that $c_h = r_h - 1$.
    \\
    This covers all possible cases for $t$ strand loss and one substitution error and the theorem is proven.
\end{proof}

\section{Conclusion}

In this work we presented a new approach to model synthesis and error corrections for composite DNA. We showed the equivalence to $t$ substitution errors with codes in the $L_1$ metric such that results of existing codes can be applied. 
Furthermore, we analysed how losses of $t$ strands are equivalent to codes in the $L_\infty$ metric. This allowed us to present a tight upper bound and present a perfect code construction. Additionally, we utilized the characteristics of the model to introduce deletion errors. An upper bound is derived for $t$-deletion correcting codes. A code construction is presented which is able to correct a single deletion and asymptotically meets the upper bound if the number of strands is odd. Finally, we mixed the error types of strand losses with substitution and presented a code construction. In future work, the model can be extended to multiple deletions and further combination of error types. 

\printbibliography

\enlargethispage{-1.2cm} 
\newpage

 \section*{Appendix}
 \vspace{-0.15cm}
\subsection{Further definitions}
\defInsertion
\defIndel

\subsection{Proof of \autoref{cla:l1-equiv}}
\proClaSubstitution

\subsection{Proof of \autoref{cla:linfty}}
\proClaStrandLoss

\subsection{Proof of \autoref{the:strand-loss-code}}
\proTheStrandLossCode

\end{document}